\pdfoutput=1
\documentclass[letterpaper, 10 pt, conference]{ieeeconf}  

\usepackage{amsmath, amsfonts, amssymb, mathtools, commath, hyperref, cancel, bm, xcolor, microtype, etoolbox, comment,mdframed}

\IEEEoverridecommandlockouts                              

\usepackage{amsmath,amsfonts,bm}

\def\eqref#1{equation~\ref{#1}}

\def\1{\bm{1}}

\DeclareMathAlphabet{\mathsfit}{\encodingdefault}{\sfdefault}{m}{sl}
\SetMathAlphabet{\mathsfit}{bold}{\encodingdefault}{\sfdefault}{bx}{n}

\usepackage{url}
\usepackage{xspace}
\usepackage{booktabs}
\usepackage{siunitx}
\usepackage{float}

\title{\LARGE \bf
Incremental stability in $p = 1$ and $p = \infty$: classification and synthesis
}

\author{Simon Kuang and Xinfan Lin$^{1}$ 
\thanks{$^{1}$ Department of Mechanical and Aerospace Engineering, University of California, Davis; Davis, CA 95616, USA. {\ttfamily \{slku, lxflin\}@ucdavis.edu}}}%

\usepackage{algorithm,algpseudocode, algorithmicx}
\usepackage{doi, xcolor}
\newtheorem{definition}{Definition}
\newtheorem{proposition}{Proposition}

\newtheorem{theorem}{Theorem}

\newtheorem{corollary}{Corollary}
\newtheorem{remark}{Remark}
\newtheorem{lemma}{Lemma}

\begin{document}
\maketitle

\begin{abstract}
All Lipschitz dynamics with the weak infinitesimal contraction (WIC) property
can be expressed as a Lipschitz nonlinear system in proportional negative feedback---this statement, a ``structure theorem,'' is true in the \(p=1\) and \(p=\infty\) norms.
Equivalently, a Lipschitz vector field is WIC if and only if it can be written as a scalar decay plus a Lipschitz-bounded residual.
We put this theorem to use using neural networks to approximate Lipschitz functions.
This results in a map from unconstrained parameters to the set of WIC vector fields, enabling standard gradient-based training with no projections or penalty terms.
Because the induced \(1\)- and \(\infty\)-norms of a matrix reduce to row or column sums, Lipschitz certification costs only \(O(d^2)\) operations---the same order as a forward pass and appreciably cheaper than eigenvalue or semidefinite methods for the \(2\)-norm.
Numerical experiments on a planar flow-fitting task and a four-node opinion network demonstrate that the parameterization (re-)constructs contracting dynamics from trajectory data.
In a discussion of the expressiveness of non-Euclidean contraction, we prove that the set of \(2\times 2\) systems that contract in a weighted \(1\)- or \(\infty\)-norm is characterized by an eigenvalue cone, a strict subset of the Hurwitz region that quantifies the cost of moving away from the Euclidean norm.
\end{abstract}

\section{Introduction}

Neural ordinary differential equations (neural ODEs)
\(\dot x = f(x; \theta)\)
are a way to parameterize continuous-time dynamic systems.
They can be interpreted as residual networks with continuous depth \cite{haber_stable_2017,chen_neural_2018}, offering memory-efficient training via adjoint methods and a principled interface with dynamical-systems theory.
However, this expressiveness brings a well-known challenge: unconstrained vector fields can exhibit diverging trajectories, sensitivity to perturbations, and loss of well-posedness under long integration horizons \cite{kidger_neural_2022}.
Therefore, we desire to enforce nonlinear stability guarantees by construction, rather than by iterative projection or post-hoc verification.

Many a student of nonlinear dynamic systems has grieved that the dynamic system \(\dot x = f(x)\) is not necessarily stable (in any sense) if all of the eigenvalues of the Jacobian \(D f(x)\) have negative real parts for all \(x\). 
However, a stronger eigenvalue condition can restore justice to this situation.
If the Gershgorin discs of \(Df(x)\) or \(Df(x)^\intercal\) are all contained in the left half-plane for all \(x\), then the system is stable---in a sense.\footnote{
    The Gershgorin circle theorem states that every eigenvalue of a matrix \((a_{ij})_{1 \leq i, j \leq n}\) lies within at least one of the discs centered at \(a_{ii}\) with radius \(R_i = \sum_{j \neq i} |a_{ij}|\), where \(R_i\) is the sum of the absolute values of the off-diagonal entries in row \(i\).
}
It is this sense of stability that the present work seeks to apply to dynamic system synthesis.
The sense is weak infinitesimal contraction in the non-Euclidean \(1\)- and \(\infty\)-norms.

\textbf{Contraction.}
While the Lyapunov notion of stability is defined in terms of invariant sets around a given equilibrium,
contraction theory is interested in the relative behavior of infinitesimally close states.
A vector field is weakly infinitesimally contracting (WIC) if nearby states remain close to first order in some norm\cite{lohmiller_contraction_1998,davydov_perspectives_2024,bullo_contraction_2026}.
The guarantees afforded to WIC systems are anything but weak: a dichotomy theorem ensures that every bounded trajectory converges, local asymptotic stability implies global asymptotic stability, and perturbations degrade the contraction rate gracefully
(\cite{jafarpour_weak_2022}; \cite[Thms.4.3--4.4]{bullo_contraction_2026}).
Any neural ODE whose vector field is WIC by construction therefore inherits these properties for free.

In what norm?
Most work on nonlinear contraction, inspired by linear time-invariant stability (which views stable linear systems as gradient flows of a positive-definite quadratic potential), has developed sufficient conditions for contraction in the $2$-norm.
Contraction beyond the $2$-norm has been studied via weak pairings and logarithmic norm conditions \cite{davydov_non-euclidean_2022},
but there are few applications that either recognize non-Euclidean contraction in a real system or target it as a synthesis principle.

\textbf{From analysis to synthesis.}
Contraction analysis of neural networks has a substantial history.
Firing-rate models $\dot x = -x + \Phi(Ax+u)$ and Hopfield models $\dot x = -x + A \Phi(x) + u$,
examples of Lur'e systems (feedback interconnection of a stable linear system and a nonlinearity),
are known to be contractive in the \(1\)- and/or \(\infty\)-norms subject to certain sufficient conditions on the nonlinearity \(\Phi\) and the matrix $A$ \cite[Lem.~3.21]{bullo_contraction_2026}.
These results are all instances of the \emph{analysis} problem: given a system, is it contracting?

The complementary \emph{synthesis} problem---parameterize all contracting vector fields so that a learning algorithm can search over them---is largely open.
We resolve it by proving a necessary and sufficient condition: every WIC vector field in \(p\in \{1, \infty\}\) is a Lur'e system.

Our structure theorem parameterizes contractive vector fields using Lipschitz functions.
It is helpful, therefore, to have an approximately universal parameterization of Lipschitz functions.

\textbf{Lipschitz neural networks.}
The Lipschitz constant of a neural network quantifies robustness and generalization \cite{virmaux_lipschitz_2018,miyato_spectral_2018}, and has become an object of study in control theory due to its relationship to (discrete-time) incremental quadratic inequalities.
The \textbf{trivial bound}, equal to the product of the Lipschitz constants of the layers, is often conservative.
Tighter bounds can be obtained via semidefinite relaxations \cite{fazlyab_efficient_2019,xu_eclipse_2024,xu_eclipse-gen-local_2025}, but these methods are computationally expensive and do not scale to large networks.

Counterintuitively, the composite problem of A) training a neural network to minimize a loss function and B) certifying that the network is Lipschitz continuous is easier than certifying an arbitrary neural network.
A recent work \cite{lipschitz_preprint} proposes to penalize the trivial bound directly in the loss function during training.
This modification, along with a careful choice of activation functions, leads to networks on which the trivial bound is tight.

\textbf{Contributions.}
In this paper we provide an \emph{unconstrained} parameterization of contracting neural ODEs for the $1$- and $\infty$-norms.
Our contributions are as follows.
\begin{enumerate}
    \item \textbf{Structure theorem (Theorem~\ref{thm:wic_decomposition}).}
    We prove an if-and-only-if condition for a WIC vector field in the $p$-norm ($p \in \{1,\infty\}$).
    The problem of parameterizing contracting vector fields reduces losslessly to parameterizing Lipschitz-bounded maps.
    This decomposition extends to weighted norms $\|Wx\|_p$ (Corollary~\ref{cor:weighted_wic_decomposition}).
    \item \textbf{Unconstrained training.}
    Combining the structure theorem with a Lipschitz-bounded feedforward network, we obtain a map from an unconstrained parameter space to the set of WIC vector fields, enabling standard gradient-based optimization with no projections or penalty terms.

    An upshot of using \(p \in \{1, \infty\}\) is that the Lipschitz estimation is of subdominant computational complexity: \(O(d^2)\) where \(d\) is the dimension of the largest matrix.
    This is an appreciable advantage compared to \(p = 2\), whereas the trivial bound takes \(O(d^3)\) and semidefinite programming roughly \(O(d^{3.5})\) (depending on the problem and the model of computation).

    \item \textbf{Eigenvalue cone condition (Theorem~\ref{thm:eigenvalue_cone}).}
    In the discussion we prove that a diagonalizable $2 \times 2$ matrix is WIC in a weighted $p$-norm ($p \in \{1,\infty\}$) if and only if its eigenvalues lie in the cone $\{\alpha + \beta i : \alpha < 0,\; |\beta| \leq -\alpha\}$, a strict subset of the Hurwitz region.
    This characterization quantifies the intrinsic conservatism of non-Euclidean contraction relative to $2$-norm contraction.
\end{enumerate}

\paragraph*{Notation}
For $x \in \mathbb{R}^n$ and $p \in [1, \infty]$, the $\ell_p$ norm is defined as
\begin{align}
    \left\|x\right\|_p = \begin{cases}
        \left(\sum_{i=1}^n |x_i|^p\right)^{1/p} & \text{if } p < \infty, \\
        \max_{i \in \{1,\ldots,n\}} |x_i| & \text{if } p = \infty.
    \end{cases}
\end{align}
For $A \in \mathbb{R}^{m \times n}$, the induced matrix norm is defined as
\begin{align}
    \left\|A\right\|_p = \sup_{x \neq 0} \frac{\left\|Ax\right\|_p}{\left\|x\right\|_p}.
\end{align}
The Jacobian of a function $f: \mathbb{R}^n \to \mathbb{R}^m$ is denoted by $D f(x) \in \mathbb{R}^{m \times n}$.
The identity function is denoted by $\mathrm{id}: \mathbb{R}^n \to \mathbb{R}^n$.
The class of Lipschitz functions \(\mathbb R^n \to \mathbb R^n\) is denoted by \(\mathrm{Lip}(\mathbb R^n)\).

\section{Background}
The following definitions and basic results on contraction theory are taken from \cite[Ch.~3--4]{bullo_contraction_2026}, modulo some simplifications such as using \(\mathbb R^n\) as the domain instead of an arbitrary convex subset.

\begin{definition}[Matrix Measure]
\label{def:matrix_measure}
Given a matrix norm $\|\cdot\|_p$, the matrix measure (or logarithmic norm) $\mu: \mathbb{R}^{n \times n} \to \mathbb{R}$ is defined as
\begin{align}
    \mu_p(A) = \lim_{h \to 0^+} \frac{\|I + hA\|_p - 1}{h}.
\end{align}
\end{definition}

\begin{definition}[{\cite[Def.~4.2]{bullo_contraction_2026}}]
\label{def:wic}
A dynamical system $\dot{x} = f(x)$ (respectively, a vector field \(f\)) is weakly infinitesimally contracting (WIC) with respect to $p$-norm if
\begin{align}
    \sup_{x \in \mathbb{R}^n} \mu_p(D f(x)) \leq 0.
\end{align}
\end{definition}

\begin{theorem}[{\cite[Thm.~3.7]{bullo_contraction_2026}}]
\label{thm:trajectory_contraction}
    Suppose that a dynamic system \(\dot x = f(x)\) is WIC with respect to the \(p\)-norm.
    Then for any trajectories \(x_1(t), x_2(t)\) of the ODE \(\dot x = f(x)\),
    \begin{align*}
        \left\|x_1(t) - x_2(t)\right\|
        &\leq
        \left\|x_1(0) - x_2(0)\right\|
    \end{align*}
    for all \(t \geq 0\).
\end{theorem}

The criteria for \(f\) to be WIC are abstract.
In the next section, we develop the theory for how to enforce this condition when \(f\) is a neural network.

\section{Enforcing Contraction}
The abstract definition of matrix measure leads to explicit formulas in certain \(p\).

\begin{lemma}[Matrix Measure for Common Norms]
\label{lem:matrix_measure_formulas}
The matrix measure for $p \in \{1, 2, \infty\}$ is given by:
\begin{align}
    \mu_1(A) &= \max_{j} \left( a_{jj} + \sum_{i \neq j} |a_{ij}| \right), \\
    \mu_2(A) &= \lambda_{\max}\left( \frac{A + A^\top}{2} \right), \\
    \mu_\infty(A) &= \max_{i} \left( a_{ii} + \sum_{j \neq i} |a_{ij}| \right).
\end{align}
\end{lemma}

Note the resemblance of the \(\{1, \infty\}\)-matrix measures to the \(\{1, \infty\}\)-matrix norms (respectively).
\begin{lemma}[Matrix Norms for Common Norms]
\label{lem:matrix_norm_formulas}
The induced matrix norm for $p \in \{1, 2, \infty\}$ is given by:
\begin{align}
    \|A\|_1 &= \max_{j} \left( |a_{jj}| + \sum_{i \neq j} |a_{ij}| \right), \\
    \|A\|_2 &= \sigma_{\max}(A), \\
    \|A\|_\infty &= \max_{i} \left( |a_{ii}| + \sum_{j \neq i} |a_{ij}| \right),
\end{align}
where $\sigma_{\max}(A)$ denotes the largest singular value of $A$.
\end{lemma}

Accordingly, matrix measures share many properties with matrix norms, such as subadditivity and homogeneity.
\begin{proposition}[Properties of Matrix Measure]
\label{prop:matrix_measure_properties}
Let $A, B \in \mathbb{R}^{n \times n}$ and $\alpha \in \mathbb{R}$. The matrix measure satisfies:
\begin{enumerate}
    \item $\mu_p(\alpha A) = \alpha \mu_p(A)$ for $\alpha \geq 0$,
    \item $\mu_p(A + B) \leq \mu_p(A) + \mu_p(B)$,
    \item $-\|A\|_p \leq \mu_p(A) \leq \|A\|_p$,
    \item $\mu_p(A + \alpha I) = \mu_p(A) + \alpha$,
    \item $|\mu_p(A) - \mu_p(B)| \leq \|A - B\|_p$,
    \item $\mu_p(A) \geq \max\{\operatorname{Re}(\lambda) : \lambda \in \sigma(A)\}$.
\end{enumerate}
\end{proposition}

The only difference between the matrix measures and the matrix norms in \(p \in \{1, \infty\}\) is that the matrix measures deal with the signed value of the diagonal elements, while the matrix norms take the absolute value.
However, we make the simple observation that a nonnegative number is equal to its absolute value, and therefore an \(f\) whose Jacobian always has a nonnegative diagonal satisfies
\begin{align}
    \mu_p(Df) = \|Df\|_p, \quad p \in \{1, \infty\}.
\end{align}
In such cases, the quantity \(\|Df\|_p\) has a familiar name:

\begin{definition}[Lipschitz constant]
\label{def:lipschitz_constant}
Let \(f: \mathbb{R}^n \to \mathbb{R}^n\) be a differentiable function. The Lipschitz constant of \(f\) with respect to the \(p\)-norm is defined as
\begin{align}
    \|f\|_{\mathrm{Lip}, p} = \sup_{x \in \mathbb{R}^n} \|D f(x)\|_p.
\end{align}
\end{definition}

Next, we state the main theorem of this paper, which connects matrix measure and Lipschitz constant.

\begin{theorem}[Structure Theorem for WIC]
\label{thm:wic_decomposition}
    Let \(f:\mathbb R^n \to \mathbb R^n\) be a Lipschitz vector field.
    Fix \(p \in \{1, \infty\}\).
    Then \(f\) is WIC if and only if
    \begin{align}
        f(x)
        &= -\gamma x + \phi(x),
    \end{align}
    for some Lipschitz \(\phi\) with \(\left\|\phi\right\|_{\mathrm{Lip}, p} \leq \gamma\).
\end{theorem}
\begin{proof}
    Proof of \emph{if} direction:\footnote{This is Exercise E3.8 in \cite{bullo_contraction_2026}.}
    By direct computation,
    \begin{align}
        Df(x) &= -\gamma I + D\phi(x),
    \end{align}
    and therefore
    \begin{align}
        \mu_p(Df(x)) &= \mu_p(-\gamma I + D\phi(x)) \\
        &= -\gamma + \mu_p(D\phi(x)) \\
        &\leq -\gamma + \left\|D\phi(x)\right\|_p \\
        &\leq -\gamma + \left\|\phi\right\|_{\mathrm{Lip}, p} \\
        &\leq 0.
    \end{align}
    Therefore, $f$ is WIC.

    Proof of \emph{only if} direction:
    Suppose that \(f\) is WIC.
    Define
    \begin{align}
        \gamma &= \max\left(0, -\inf_{x \in \mathbb{R}^n} \inf_{i \in [n]} \dpd{f_i}{x_i}(x)\right),
        \\
        \phi(x) &= f(x) + \gamma x.
    \end{align}
    Moreover, we know that \(\gamma < \infty\), because 
    \begin{align}
        \gamma \leq \sup_{x\in \mathbb R^n} \sup_{i, j \in [n]} |Df(x)_{ij}|,
    \end{align}
    the entry-wise norm, which is finite because \(f\) is Lipschitz and all norms on finite-dimensional spaces are equivalent.
    
    Therefore, in the case that \(p = 1\),
    \begin{align*}
        \left\|D\phi(x)\right\|_1
        &=
        \max_{i} \del{
            \left|Df(x)_{ii} + \gamma\right|
            + \sum_{j \neq i} |Df(x)_{ij}|
        }
        \\
        &=
        \max_{i} \del{
            Df(x)_{ii} + \gamma
            + \sum_{j \neq i} |Df(x)_{ij}|
        }
        \\
        &=
        \gamma + \max_{i} \del{
            Df(x)_{ii}
            + \sum_{j \neq i} |Df(x)_{ij}|
        }
        \\
        &=
        \gamma + \mu_1(Df(x))
        \\
        &\leq \gamma,
    \end{align*}
    and in the case that \(p = \infty\),
    \begin{align*}
        \left\|D\phi(x)\right\|_\infty
        &=
        \max_{j} \del{
            \left|Df(x)_{jj} + \gamma\right|
            + \sum_{i \neq j} |Df(x)_{ij}|
        }
        \\
        &=
        \max_{j} \del{
            Df(x)_{jj} + \gamma
            + \sum_{i \neq j} |Df(x)_{ij}|
        }
        \\
        &=
        \gamma + \max_{j} \del{
            Df(x)_{jj}
            + \sum_{i \neq j} |Df(x)_{ij}|
        }
        \\
        &=
        \gamma + \mu_\infty(Df(x))
        \\
        &\leq \gamma,
    \end{align*}
\end{proof}

This theorem establishes that for \(p \in \{1, \infty\}\), the map
\begin{multline*}
    (\mathbb R_{\geq 0}, \mathrm{Lip}(\mathbb R^n))
    \\
    \to \mathrm{Lip}(\mathbb R^n)
    \cap \del{\text{\(p\)-WIC functions}: \mathbb R^n \to \mathbb R^n}
\end{multline*}
defined by
\begin{gather*}
    (\epsilon, f) \mapsto (-\epsilon + \|f\|_{\mathrm{Lip}, p}) \operatorname{id} + f
\end{gather*}
is a surjection.
In other words, the problem of parameterizing Lipschitz WIC functions reduces to parameterizing Lipschitz functions.

\begin{remark}
\label{rem:lipschitz_assumption}
    It is not an onerous restriction to require WIC functions to be furthermore Lipschitz, as Lipschitz continuity is already a natural assumption in theoretical analyses which require existence and uniqueness of classical solutions to ODEs, as well as accurate numerical schemes for integrating them.
\end{remark}

\subsection{Extension to weighted norms}
An LTI system \(\dot x = Ax\) with Hurwitz \(A\) may not necessarily contract in the 2-norm \(\sqrt{x ^\intercal x}\), but
\(A\) necessarily satisfies a Lyapunov equation:
\begin{align*}
    A^\intercal P + PA < 0
\end{align*}
for some positive definite \(P > 0\).
This motivates the consideration of weighted norms \(\|Wx\|_p\) for invertible \(W\).

\begin{definition}[Weighted WIC]
\label{def:weighted_wic}
    A vector field \(f: \mathbb{R}^n \to \mathbb{R}^n\) is said to be weighted WIC (wWIC) if there exists an invertible weight matrix \(W \in \mathbb{R}^{n \times n}\) such that \(f\) is WIC with respect to the weighted norm \(\|Wx\|_p\).
\end{definition}

Applying Theorem~\ref{thm:wic_decomposition} to \(W^{-1} f\), we obtain: 
\begin{corollary}[Structure theorem for weighted norms]
\label{cor:weighted_wic_decomposition}
    Let \(f: \mathbb{R}^n \to \mathbb{R}^n\) be Lipschitz and fix \(p \in \{1, \infty\}\).
    Then \(f\) is wWIC with weight matrix \(W\) if and only if
    \begin{align}
        f(x) &= -\gamma\, x + W^{-1}\phi(Wx)
    \end{align}
    for some \(\gamma \geq 0\) and some Lipschitz \(\phi\) with \(\left\|\phi\right\|_{\mathrm{Lip},\,p} \leq \gamma\).
\end{corollary}

\section{Lipschitz by construction: neural networks}
Whereas it is not straightforward to parameterize WIC functions directly,
we can instead parameterize Lipschitz functions.
The Lipschitz-to-WIC mapping is lossless (Theorem~\ref{thm:wic_decomposition}).
In order to get a lossless parameterization of WIC functions, we need to parameterize Lipschitz functions.

Every neural network (with a Lipschitz activation function) is Lipschitz,
and Lipschitz regularity is more than sufficient for universal approximation by neural networks.
The challenge consists, however, in determining what the Lipschitz constant is.
It is NP-hard to approximate the Lipschitz constant of a general neural network \(\phi\) to arbitrary accuracy;
Theorem~\ref{thm:wic_decomposition}
demands a minimally conservative upper bound.



To obtain a minimally conservative Lipschitz upper bound,
we employ activation functions whose Jacobians always land in extreme points of the operator norm ball.
Following \cite{lipschitz_preprint}, we use \emph{saturated polyactivations}: vector-valued activation functions $\vec\sigma : \mathbb R^d \to \mathbb R^{dK}$ whose layerwise Jacobian satisfies
\begin{align*}
\left\|\dpd{}{x} \vec\sigma(x)\right\|_p = 1 \quad \text{for almost every } x \in \mathbb R^d.
\end{align*}
This ensures that the activation nonlinearity contributes unit Lipschitz constant, making the network's Lipschitz bound tight with respect to the weight matrices.
For example, the absolute value $x \mapsto |x|$ saturates all $p$-norms, while the CReLU $x \mapsto (\max(0,x), \max(0,-x))$ saturates $p \in [1,\infty)$ \cite{shang_understanding_2016,lipschitz_preprint}.

\section{Numerical Examples}
\label{sec:numerical-examples}
\subsection{Toy example ($1$-norm contraction)}
\label{sec:toy-example}
We generate $N = 20$ random pairs $\{(x_i(0),\, x_i(T))\}_{i=1}^N$ with $T = 1$;
these are interpreted as origin-destination pairs of an unknown flow.
We hope to find the ``best-fit'' 1-norm w-WIC vector field that explains all of these pairs.

To do so, we train a contractive neural ODE of the form
\begin{align*}
    \dot x = -\|\phi(\cdot;\theta)\|_{\mathrm{Lip},1} x + W\phi(W^{-1} x;\theta),
\end{align*}
by minimizing (in the norm weight \(W\) and the network parameters \(\theta\)) the profile log-likelihood loss
\begin{align*}
    \mathcal L(\theta, W) = \det\del{
        \frac{1}{N}\sum_{i=1}^N r_i r_i^\top
    },
    \quad r_i = \hat x_i(T;\theta) - x_i(T),
\end{align*}
where $\hat x_i(T;\theta)$ is the neural ODE prediction from $x_i(0)$.
The map $\phi$ is a one-hidden-layer MLP with 40 hidden units, vector-valued activation $x \mapsto \del{\frac{x + \sin x}{2}, \frac{x - \sin x}{2}}$.
Training uses 400 steps of Cocob \cite{orabona_training_2017}.

\begin{figure*}
    \centering
    \includegraphics[width=\linewidth]{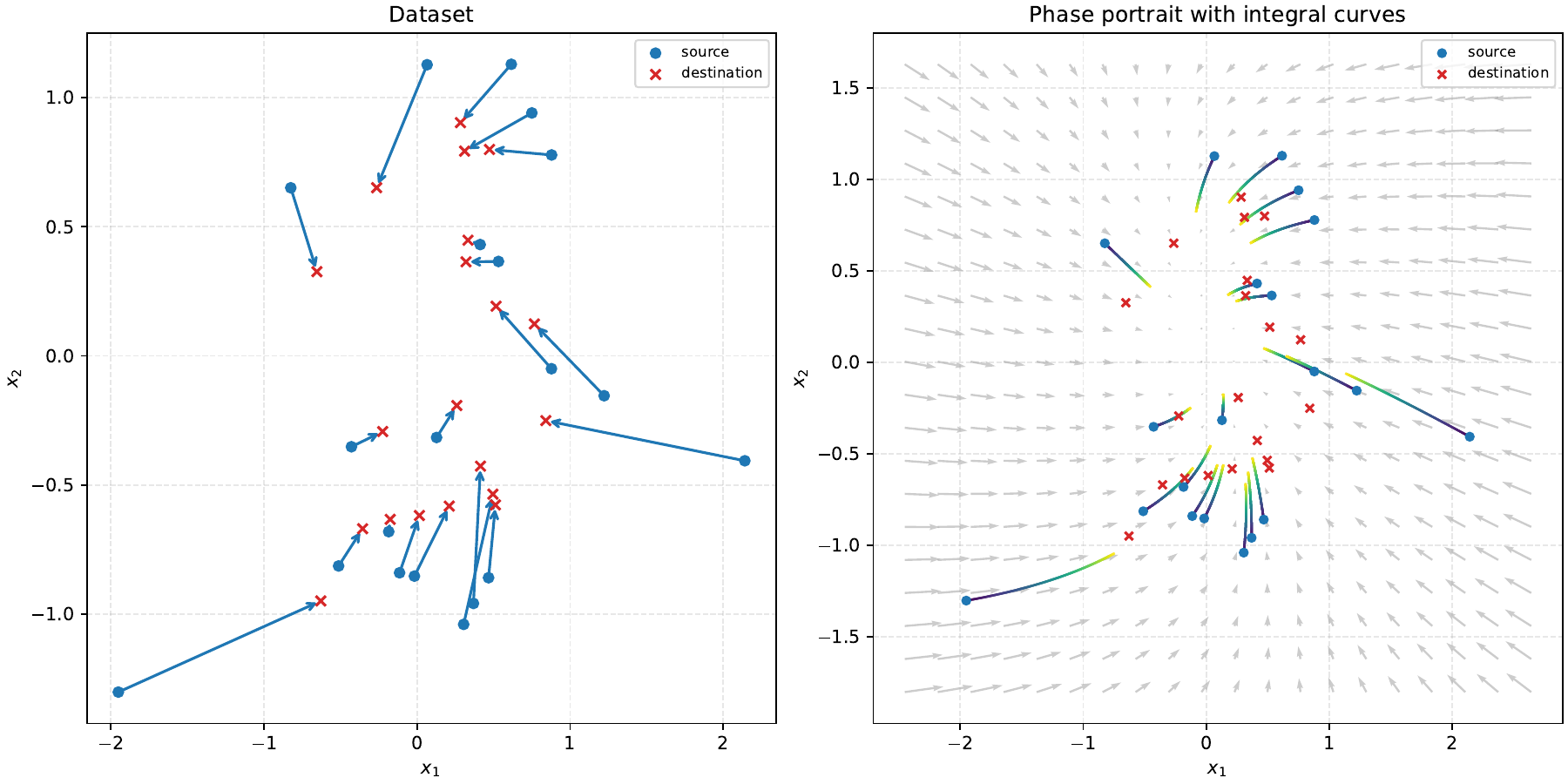}
    \caption{Phase portrait comparison for the toy example. Left: origin-destination pairs. Right: learned contractive neural ODE.}
    \label{fig:toy_example}
\end{figure*}


\subsection{Opinion network ($\infty$-norm contraction)}
\label{sec:opinion-network}
The ground truth system is a nonlinear opinion dynamics model on a strongly connected weighted digraph with $n = 4$ nodes:
\begin{align}
    \dot x = -D\,x + \nu\,\tanh(A\,x),
    \label{eq:opinion_dynamics}
\end{align}
where $A \in \mathbb R^{4 \times 4}_{\geq 0}$ is the adjacency matrix, $D = \operatorname{diag}(A\mathbf 1_4)$ is the out-degree matrix, and $\nu \in [0,1]$ is an activation parameter.
Exercise 4.5 \cite{bullo_contraction_2026} shows that this system, derived from \cite{gray_multiagent_2018}, is WIC in the $\infty$-norm for all $\nu \leq 1$.
We set $\nu = 1$ (weak contraction, $\mu_\infty = 0$).

We generate $N = 100$ training and $50$ test trajectory pairs $(x_i(0),\, x_i(T))$ with $T = 2$ from initial conditions $x_i(0) \sim \mathcal N(0, 4I_4)$.
A contractive neural ODE with $p = \infty$, a one-hidden-layer MLP with 40 hidden units, and $x \mapsto (|x|)$ activation is trained for 2000 steps.

\begin{figure*}
    \centering
    \includegraphics[width=\linewidth]{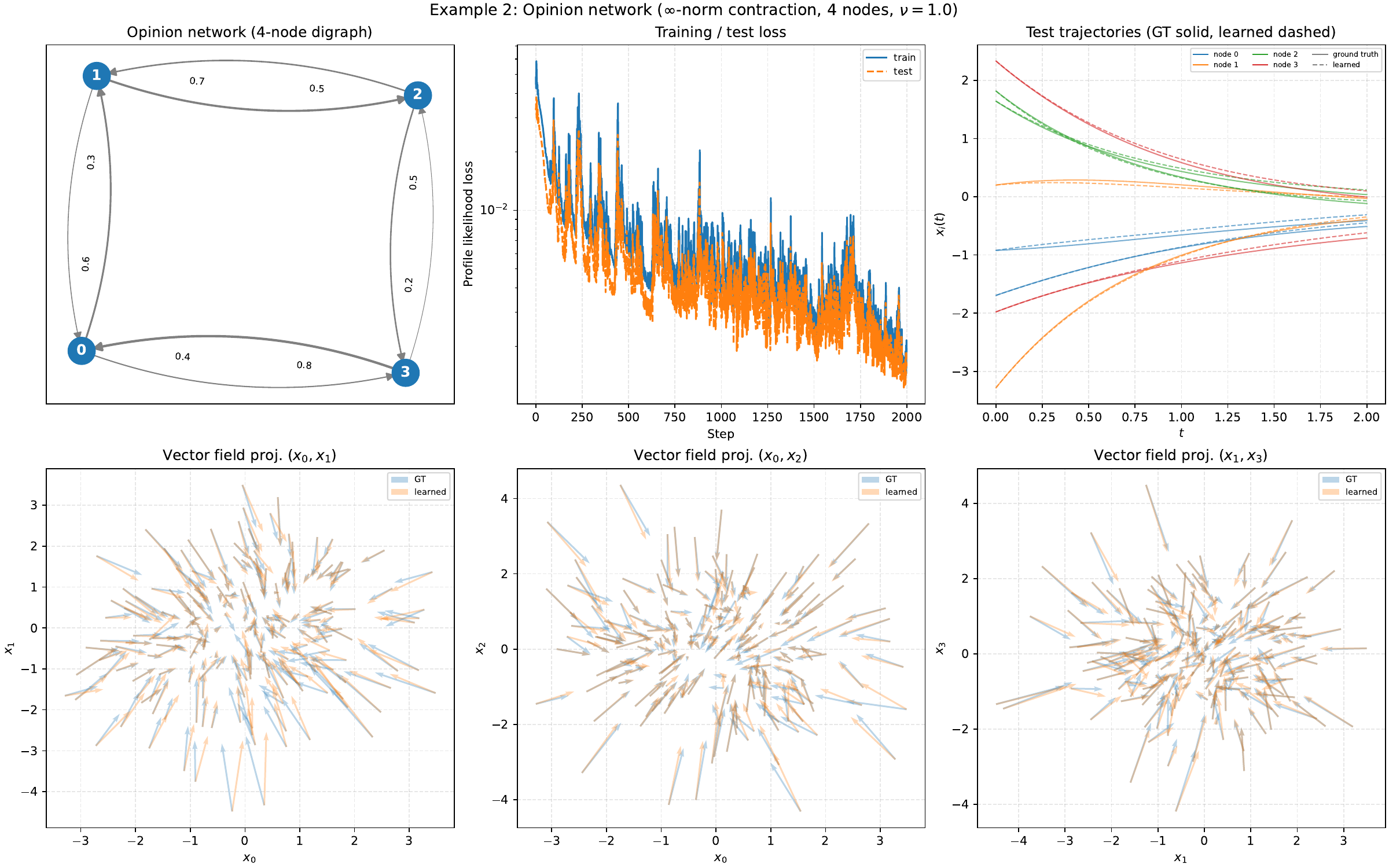}
    \caption{Opinion network example.
    Top left: the 4-node weighted digraph defining $A$.
    Top middle: training and test loss.
    Top right: ground truth (solid) and learned predictions (dashed) for one initial condition in the test dataset.
    Bottom row: pairwise projections of the ground-truth (blue) and learned (orange) vector fields onto the $(x_0,x_1)$, $(x_0,x_2)$, and $(x_1,x_3)$ planes.}
    \label{fig:opinion_network}
\end{figure*}

Figure~\ref{fig:opinion_network} shows the training and test loss curves, ground-truth and learned time series for six initial conditions, and pairwise projections of the vector fields.
The learned model captures the consensus-seeking dynamics of \eqref{eq:opinion_dynamics}: trajectories converge toward the origin, reproducing the qualitative behavior of the ground-truth system while maintaining $\infty$-norm contractivity by construction.

\section{Discussion: where are all the \texorpdfstring{\(p \in \{1, \infty\}\)}{p in {1, infinity}}-contractions?}
We begin by praising the advantages of incremental contraction in \(p\in \{1,\infty\}\) before reflecting on their limitations.

\begin{enumerate}
    \item The structure theorem \ref{thm:wic_decomposition} and its corollary for weighted norms comprise an if-and-only-if description of \(1\)- and \(\infty\)-norm WIC systems in terms of a simpler concept (Lipschitz functions).
    \item Verifying the Lipschitz constant in \(p \in \{1, \infty\}\) only requires \(O(d^2)\) operations, where \(d\) is the state or hidden layer dimension, whichever is higher.
    This is thus of the same order as feedforward evaluation.
    Eigenvalue and semidefinite programming techniques for the \(p=2\) norm are \(O(d^3)\) and \(O(d^{3.5})\), respectively.
\end{enumerate}

Why, then is the literature on systems \(\dot x = f(x)\) that contract in the the \(p\)-norm for \(p \in \{1, \infty\}\) rich in theory and poor in concrete examples?
The ratio of \(\textbf{definitions}:\textbf{theorems}:\textbf{examples}\) in a theory paper on continuous-time contraction in the 2-norm is \(2:2:9\).\footnote{Here, ``theorems'' includes lemmas and ``examples'' includes counterexamples.}
A recent work on non-Euclidean contraction theory has a ratio of \(14:34:1\).
We posit that this literature (the current work included) is top-heavy because \(p \in \{1, \infty\}\)-contractions are intrinsically less diverse and less physical---unphysical should not, however be taken to mean uninteresting or unimportant, and it does not diminish the mathematical and substantive contributions that have arised in the analysis of physically motivated Lur'e systems (such as continuous-time neuron models), synthetic applications such as opinion dynamics, and others.

Our structure theorem explains why most contractive systems (\(p \in \{1, \infty\}\)) in the analysis literature are Lur'e systems.
It is because in fact, all of them are.

\subsection{Diversity}
Recall that infinitesimal contraction means that along a trajectory
\(\dot x = f(x)\),
the variational linear time-varying equation
\begin{align}
    \delta \dot x &= \underbrace{\frac{\partial f}{\partial x}(x)}_{:= A} \delta x
    \label{eq:variational}
\end{align}
satisfies the contraction condition.
Over short times \(\Delta t\),
\begin{align*}
    \delta x(t + \Delta t) &\approx \exp(A \Delta t) \delta x(t)
\end{align*}
and the \(p\)-norm contraction condition is satisfied if \(\left\|\exp(A \Delta t)\right\|_p < 1\).
How many such matrices are there for each \(p\)?

As a proxy for the diversity of contracting flows, we may assess the size of the isometry group of the \(p\)-norm, namely those matrices \(O \in \mathbb{R}^{n \times n}\) satisfying \(\|Ox\|_p = \|x\|_p\) for all \(x \in \mathbb{R}^n\).
In \(p=2\), the isometry group is the orthogonal group \(O(n)\), which is a smooth manifold of dimension \(\frac{n(n-1)}{2}\).
In \(p \in [1, \infty] - \{2\}\), the isometry group is discrete and consists only of signed permutation matrices.
This is essentially because the unit sphere is a (rounded) polyhedron, and any isometry must map vertices to vertices.\footnote{This fact was communicated to the first author by Ken Brown and Gavin Pandya.}

\subsection{Physicality}
The laws of mechanics, as well as those of electricity and magnetism, are symmetric with respect to solid-body rotation.
This means that the potential and kinetic energy scalars from which laws of motion are derived in the Lagrangian or Hamiltonian formalisms can feature no \(p\)-norm other than \(p=2\).
In contrast, \(p \in \{1, \infty\}\) norms do not arise naturally in mechanics.
Many mechanical systems which are stable (in the energy sense) are not contractive in the \(p \in \{1, \infty\}\) sense.

To make this sense of strictness precise, let us identify the generic contracting linear systems \(\dot x = Ax\) in dimension 2.
According to modal decomposition, for \(\dot x= Ax\) to be WIC in a weighted \(2\)-norm, it is necessary and sufficient for \(A\) to marginally stable.
In the \(p=1\) and \(p=\infty\) cases, the eigenvalues must lie in a strict subset of the left half-plane.

\begin{theorem}
\label{thm:eigenvalue_cone}
    Let \(A \in \mathbb R^{2 \times 2}\) be diagonalizable and \(p \in \{1, \infty\}\). Then \(A\) is WIC in a weighted \(p\)-norm if and only if its eigenvalues lie in the cone \(\{\alpha + \beta i : \alpha \leq 0, |\beta| \leq -\alpha\}\).
\end{theorem}
\begin{proof}
    \emph{If direction.}
    Suppose the eigenvalues of \(A\) lie in the cone.
    If \(A\) has real eigenvalues \(\lambda_1, \lambda_2 \leq 0\), then \(A\) is similar to \(\operatorname{diag}(\lambda_1, \lambda_2)\), which has \(\mu_p = \max(\lambda_1, \lambda_2) \leq 0\).
    If \(A\) has complex eigenvalues \(\alpha \pm \beta i\) with \(|\beta| \leq -\alpha\), then \(A\) is similar to \(B = \begin{psmallmatrix} \alpha & \beta \\ -\beta & \alpha \end{psmallmatrix}\), and Lemma~\ref{lem:matrix_measure_formulas} gives \(\mu_1(B) = \alpha + |\beta| \leq 0\).

    \emph{Only if direction.}
    WIC of \(\dot x = Ax\) in \(\|Wx\|_p\) is equivalent to \(\mu_p(C) \leq 0\) where \(C = WAW^{-1}\).
    By Proposition~\ref{prop:matrix_measure_properties}(6), every eigenvalue has \(\operatorname{Re}(\lambda) \leq 0\).
    It remains to show that complex eigenvalues \(\alpha \pm \beta i\) satisfy \(|\beta| \leq -\alpha\).
    Let \(C \in \mathbb R^{2 \times 2}\) satisfy \(\mu_p(C) \leq 0\) with eigenvalues \(\alpha \pm \beta i\), \(\beta \neq 0\).
    The trace and determinant of \(C\) are determined by the eigenvalues:
    \begin{align*}
        c_{11} + c_{22} &= 2\alpha, \\
        c_{11}c_{22} - c_{12}c_{21} &= \alpha^2 + \beta^2.
    \end{align*}
    For \(p = 1\), the matrix measure is
    \begin{align*}
        \mu_1(C) = \max(c_{11} + |c_{21}|,\; c_{22} + |c_{12}|).
    \end{align*}
    Thus \(\mu_1(C) \leq 0\) requires \(c_{11} \leq -|c_{21}|\) and \(c_{22} \leq -|c_{12}|\), hence both diagonal entries are nonpositive and
    \begin{align*}
        |c_{12}|\,|c_{21}| \leq |c_{11}|\,|c_{22}| = c_{11} c_{22}.
    \end{align*}
    Then
    \begin{align*}
        \alpha^2 + \beta^2 = c_{11}c_{22} - c_{12}c_{21} \leq c_{11}c_{22} + |c_{12}c_{21}| \leq 2c_{11}c_{22}.
    \end{align*}
    By AM--GM on the nonpositive diagonal entries,
    \begin{align*}
        c_{11}c_{22} \leq \del{
            \frac{c_{11}+c_{22}}{2}
        }^2 = \alpha^2,
    \end{align*}
    so \(\alpha^2 + \beta^2 \leq 2\alpha^2\), i.e.\ \(|\beta| \leq |\alpha| = -\alpha\).
    The \(p = \infty\) case can be obtained by transposing \(C\), which swaps the roles of rows and columns while preserving the trace and determinant.
\end{proof}

\begin{remark}
\label{rem:trace_det_region}
    Theorem~\ref{thm:eigenvalue_cone} can be restated in terms of the similarity invariants \(\tau = \operatorname{tr} A\) and \(\delta = \det A\).
    For \(p = 2\), WIC in a weighted \(2\)-norm requires only that \(A\) be marginally stable, i.e.\ \(\tau < 0\) and \(\delta > 0\).
    For \(p \in \{1, \infty\}\), the eigenvalue cone imposes the additional constraint \(\delta \leq \tau^2/2\).

    To see this, note that if \(A\) has real eigenvalues \(\lambda_1, \lambda_2 \leq 0\), then \(\tau = \lambda_1 + \lambda_2 \leq 0\), \(\delta = \lambda_1 \lambda_2 \geq 0\), and the discriminant condition gives \(\delta \leq \tau^2/4 \leq \tau^2/2\).
    If \(A\) has complex eigenvalues \(\alpha \pm \beta i\), then \(\tau = 2\alpha\) and \(\delta = \alpha^2 + \beta^2\).
    The cone condition \(|\beta| \leq -\alpha\) becomes \(\delta = \alpha^2 + \beta^2 \leq 2\alpha^2 = \tau^2/2\).

    Thus, for \(p \in \{1, \infty\}\), the WIC region in the \((\tau, \delta)\)-plane is
    \begin{align*}
        \{\,(\tau, \delta) : \tau \leq 0,\; 0 < \delta \leq \tau^2/2\,\},
    \end{align*}
    a strict subset of the stable region \(\{\tau \leq 0,\; \delta \geq 0\}\); see Figure~\ref{fig:trace_det}.
    The boundary parabola \(\delta = \tau^2/2\) corresponds to eigenvalues on the rays \(\alpha \pm \alpha i\) (\(\alpha < 0\)), i.e.\ the boundary of the cone at \(45^\circ\) from the negative real axis.
\end{remark}

\begin{figure}
    \centering
    \includegraphics[width=\columnwidth]{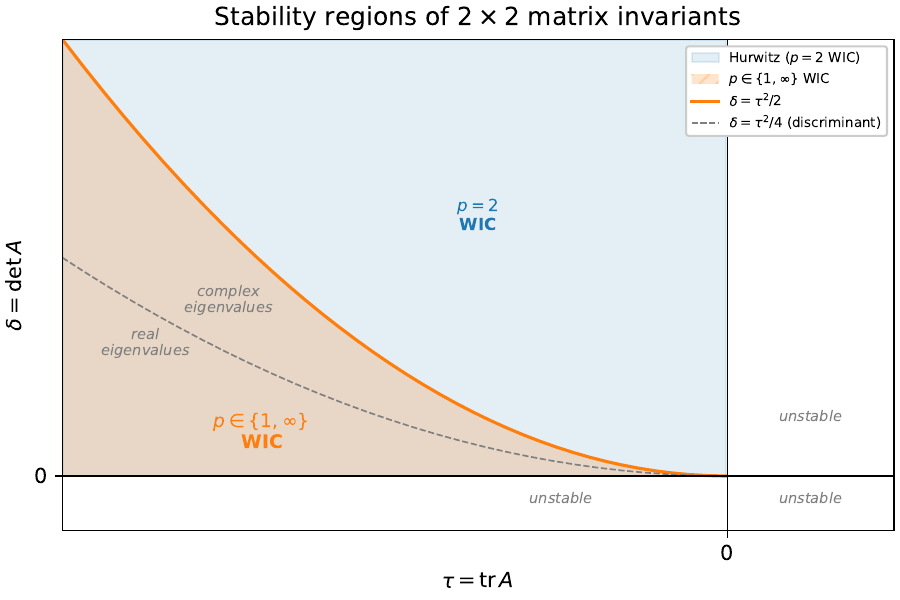}
    \caption{The $(\tau,\delta)$-plane for diagonalizable $2\times 2$ matrices with $\tau = \operatorname{tr} A$ and $\delta = \det A$.
    The shaded marginally stable region ($\tau < 0$, $\delta > 0$) characterizes $p = 2$ WIC.
    The hatched subregion below the parabola $\delta = \tau^2/2$ characterizes $p \in \{1,\infty\}$ WIC.
    The dashed curve $\delta = \tau^2/4$ is the zero set of the discriminant $\tau^2 - 4\delta$: eigenvalues are real below it and complex above it.}
    \label{fig:trace_det}
\end{figure}



\section{Conclusion}

We presented an unconstrained parameterization of neural ODEs that are weakly infinitesimally contracting in the \(1\)- and \(\infty\)-norms.
The key ingredient is a structure theorem (Theorem~\ref{thm:wic_decomposition}) that reduces the design of contracting vector fields to the design of Lipschitz-bounded maps, a problem for which scalable neural-network architectures already exist.
The resulting training procedure requires no projections, barrier functions, or penalty terms, and the per-layer Lipschitz certification costs only \(O(d^2)\) operations.
Preliminary numerical experiments on a toy flow-fitting task and a four-node opinion network demonstrate that the parameterization can recover contracting dynamics from trajectory data.

An eigenvalue-cone characterization (Theorem~\ref{thm:eigenvalue_cone}) clarifies the price of non-Euclidean contraction: in dimension two, the set of \(p \in \{1,\infty\}\)-contracting linear systems is a strict subset of the marginally stable systems, reflecting a fundamental trade-off between computational tractability and expressiveness.
Extending this characterization to higher dimensions and to other values of \(p\) remains an open problem.

\bibliographystyle{plain}
\bibliography{neural-ode-contraction,lipschitz-preprint}

\end{document}